\documentclass[11pt]{article}
\usepackage[utf8]{inputenc}
\usepackage[a4paper, total={6in, 10.5in}]{geometry}
\usepackage{amssymb}
\usepackage{amsthm}
\usepackage{amsmath}
\usepackage{mathtools}
\usepackage{bbm}
\usepackage{algpseudocode}
\usepackage[nodayofweek,level]{datetime}
\usepackage{graphicx}
\usepackage{multicol}
\usepackage{caption}
\usepackage{subcaption}
\usepackage[T5,T1]{fontenc}
\title{Probabilistic Counting in Generalized Turnstile Models\thanks{This work was supported by NSF grant CCF-2221980. }}
\author{Dingyu Wang\\University of Michigan\\{wangdy@umich.edu}}
\date{}

\usepackage{amssymb}
\usepackage{amsthm}
\usepackage{amsmath}
\usepackage{mathtools}
\usepackage{bbm}
\usepackage{wasysym}

\theoremstyle{plain}
\newtheorem{theorem}{Theorem}
\newtheorem{lemma}{Lemma}

\theoremstyle{definition}
\newtheorem{definition}{Definition}
\newtheorem{assumption}{Assumption}

\theoremstyle{remark}
\newtheorem{remark}{Remark}

\DeclarePairedDelimiter{\norm}{\|}{\|}

\newcommand{\ind}[1]{\mathbbm{1}\left[#1\right]}
\newcommand{\var}{\mathbb{V}}
\newcommand{\pr}{\mathbb{P}}
\newcommand{\E}{\mathbb{E}}

\newcommand{\R}{\mathbb{R}}

\newcommand{\Z}{\mathbb{Z}}

\newcommand{\N}{\mathbb{N}}
\newcommand{\F}{\mathbb{F}}

\newcommand{\pcsa}{\textsf{PCSA}}
\newcommand{\knw}{\textsf{KNW}}
\newcommand{\bitmap}{\textrm{BITMAP}}
\newcommand{\fieldmap}{\textrm{FIELDMAP}}
\newcommand{\distcount}{\textrm{Distinct-Count}}

\begin{document}
\maketitle

\begin{abstract}
Traditionally in the turnstile model of data streams, there is a state vector $x=(x_1,x_2,\ldots,x_n)$ which is updated through a stream of pairs $(i,k)$ where $i\in [n]$ and $k\in \Z$. Upon receiving $(i,k)$, $x_i\gets x_i + k$. A distinct count algorithm in the turnstile model takes one pass of the stream and then estimates $\norm{x}_0 = |\{i\in[n]\mid x_i\neq 0\}|$ (aka $L_0$, the Hamming norm).

In this paper, we define a finite-field version of the turnstile model. Let $F$ be any finite field. Then in the $F$-turnstile model,  for each $i\in [n]$, $x_i\in F$; for each update $(i,k)$, $k\in F$. The update $x_i\gets x_i+k$ is then computed in the field $F$. A distinct count algorithm in the $F$-turnstile model takes one pass of the stream and estimates $\norm{x}_{0;F} = |\{i\in[n]\mid x_i\neq 0_F\}|$. 

We present a simple distinct count algorithm, called $F$-\pcsa{}, in the $F$-turnstile model for any finite field $F$. The new $F$-\pcsa{} algorithm takes $m\log(n)\log (|F|)$ bits of memory and estimates $\norm{x}_{0;F}$ with $O(\frac{1}{\sqrt{m}})$ relative  error where the hidden constant depends on the order of the field.

$F$-\pcsa{} is straightforward to implement and has several applications in the real world with different choices of $F$. Most notably, 
\begin{itemize}
    \item $\F_2$-\pcsa{} can be used to estimate the number of non-zero elements in a boolean system where every element has a boolean state which can be toggled during a stream of updates;
    \item $\F_{2^k}$-\pcsa{} can be used to estimate the number of non-zero elements in a $k$-boolean system where every element has a vector of $k$ boolean attributes which can be individually toggled during a stream of updates;
    \item $\F_p$-\pcsa{} can be used to implement simple and efficient $L_0$-estimation algorithms in the (original) turnstile model, making the probabilistic counting with deletions conceptually no harder than the probabilistic counting without deletions. With different schemes to select $p$, one gets different upper bounds. To get a $(1\pm \epsilon)$-estimation of $L_0$ with 1/3 error probability, one needs
    \begin{itemize}
    \item $O(\epsilon^{-2}\log n \cdot (\log \epsilon^{-1}+\log\log \norm{x}_\infty))$ bits, matching the state of the art $L_0$-estimation algorithm by Kane, Nelson and Woodruff \cite{kane2010optimal};
    \item $O(\epsilon^{-2}\log n \cdot \log \norm{x}_\infty)$ bits, which is useful when the maximum frequency is small;
    \item $O(\epsilon^{-2}\log n \cdot \max(\log \frac{\norm{x}_1}{\norm{x}_0}, \log \epsilon^{-1})) $ bits, which is useful when the average non-zero frequency is low.
\end{itemize}
The upper bounds above can be further sharpened by ignoring the top $O(\epsilon)$-fraction of elements with non-zero frequencies.
\end{itemize}

\end{abstract}

\section{Introduction}
\emph{\distcount} is the problem of approximately counting the number of distinct elements in a stream. In 1983, Flajolet and Martin \cite{flajolet1985probabilistic} proposed the first streaming \distcount{} algorithm, Probabilistic Counting with Stochastic Averaging (\pcsa{}).  Although being the first streaming \distcount{} algorithm in history, the \pcsa{} sketch already possesses two fundamental characteristics that are desirable when counting distinct elements.
\begin{itemize}
    \item ``Zero knowledge''. The distribution of \pcsa{}'s memory state depends only on the cardinality of the stream. 
    \item ``Scale invariance''. The behavior of the sketch is multiplicatively periodic in the cardinality. 
\end{itemize}

However, \pcsa{} only works in the \emph{incremental} setting, in which only insertions are allowed.

\begin{definition}[incremental setting]
    Let $U=\{1,2,\ldots,n\}$ be the universe (e.g.~the set of IPv4 addresses). 
    The state vector $x=(x_1,x_2,\ldots,x_n)\in \N^n$ is initialized as all zeros and gets updated by a stream of pairs in the form of $(v,k)$, where $v\in U$ and $k\in \N$.
    \begin{itemize}
        \item Upon receiving $(v,k)$, $x_v \gets x_v +k$.
    \end{itemize}   
    A streaming \distcount{} algorithm in the incremental setting takes one pass of the stream and estimates $\norm{x}_0=|\{v\in U:x_v\neq 0\}|$ (aka cardinality, $F_0$), i.e., the number of elements that have non-zero frequency. 
\end{definition}
\begin{remark}
    By convention (e.g.~\cite{alon1996space}), we refer to $x_v$ as the \emph{frequency} of $v$ when $x_v$ encodes the number of occurrences of $v$. However, for a general-purpose $x_v$, we will call it the \emph{state} of $v$. 
\end{remark}

In 2003, Cormode, Datar, Indyk, and Muthukrishnan \cite{cormode2003comparing} gave the first \distcount{} algorithm in the \emph{turnstile model}, allowing both insertions and deletions.
\begin{definition}[turnstile model, \cite{cormode2003comparing}, \cite{li2014turnstile}]
    Let $U=\{1,2,\ldots,n\}$ be the universe. 
    The state vector $x=(x_1,x_2,\ldots,x_n)\in \Z^n$ is initialized as all zeros and gets updated by a stream of pairs in the form of $(v,k)$, where $v\in U$ and $k\in \Z$.
    \begin{itemize}
        \item Upon receiving $(v,k)$, $x_v \gets x_v +k$.
    \end{itemize}
    A streaming \distcount{} algorithm in the turnstile model takes one pass of the stream and estimates $\norm{x}_0=|\{v\in U:x_v\neq 0\}|$ (aka $L_0$ or the Hamming norm).

\end{definition}
\begin{remark}
    One may also define the updates as $\mathsf{Insert}(v)$ and $\mathsf{Delete}(v)$ in the turnstile model, which is equivalent to receiving $(v,+1)$ and $(v,-1)$ respectively.
\end{remark}
\begin{remark}
    Algorithms in the turnstile model natively support unions and symmetric differences of streams, since they are often \emph{linear} \cite{li2014turnstile}.
\end{remark}

\paragraph{The key insight.} The application of a streaming model is intimately related to the \emph{algebraic structure} of the update ($k$) and the state ($x_v$). The only difference in the definitions of the incremental setting and the turnstile model is that the former assumes $k,x_v\in \N$, while the latter assumes $k,x_v\in \Z$. The fact that $\Z$ has an additive inverse allows the turnstile model to incorporate deletions natively. With this observation, it is natural to define the following general model as a framework.

\begin{definition}[$M$-turnstile model]\label{def:fturnstile}
Let $(M,+)$ be a commutative monoid. The \emph{$M$-turnstile model} is defined as follows. 
     Let $U=\{1,2,\ldots,n\}$ be the universe. 
    The state vector $x=(x_1,x_2,\ldots,x_n)\in M^n$ is initialized as all zeros (i.e.~the identity of the monoid $M$) and gets updated by a stream of pairs in the form of $(v,k)$, where $v\in U$ and $k\in M$.
    \begin{itemize}
        \item Upon receiving $(v,k)$, $x_v \gets x_v +k$.
    \end{itemize}
    A streaming \distcount{} algorithm in the $M$-turnstile model takes one pass of the stream and estimates $\norm{x}_{0}=|\{v\in U:x_v\neq 0\}|$. 
\end{definition}
\begin{remark}
    With this definition, the incremental setting is just the $\N$-turnstile model and the original turnstile model is the $\Z$-turnstile model. In general, $M$ is the structure of the state of the elements. For example, if every element has a boolean state (on/off), then $M$ is the $\{0,1\}$-field. 
See Table \ref{tab:model_framework} for a list of models defined in the $M$-turnstile framework.    
\end{remark}
\begin{remark}
In order for the state vector $x\in M^n$ to be updated in parallel, the algebraic structure of $M$ must at least be a commutative monoid. Specifically, associativity and commutativity are needed to merge streams unorderedly, and the existence of identity is needed to initialize the values of all elements.  There are similar notions in recent work in the field of automata theory \cite{alur2017automata} and commutative algebra \cite{baez2019network}. 
\end{remark}
\begin{table}[ht]
    \centering
    \begin{tabular}{c|c | c}
       model name & $M$-turnstile framework & note \\
        \hline
       incremental setting & $\N$-turnstile model\\
      (original) turnstile model & $\Z$-turnstile model\\
       $k$-boolean model $\dagger\star$  & $\F_{2^k}$-turnstile model & $k$ is a positive integer\\
       $p$-cyclic model $\dagger\star$  & $\F_{p}$-turnstile model & $p$ is a prime\\
       finite field model $\dagger\star$  & $F$-turnstile model & $F$ is a finite field\\ 
       invertible model $\dagger$ & $G$-turnstile model & $G$ is an abelian group \\ 
       base model $\dagger$ & $M$-turnstile model & $M$ is a commutative monoid
    \end{tabular}
    \caption{Examples of models in the $M$-turnstile framework. $\dagger$: The model is newly defined in this work. $\star$: The \distcount{} problem in the model is solved in this work by $F$-\pcsa{}.}
    \label{tab:model_framework}
\end{table}

In general, it is not yet clear how to generally solve \distcount{} in the $M$-turnstile model for an arbitrary commutative monoid $M$. However, we show that when $M$ is a \emph{finite field} $F$, it is straightforward to solve \distcount{} in the $F$-turnstile model, conceptually no harder than solving \distcount{} in the incremental setting. 

\paragraph{Contribution.}
Concretely, we present a simple streaming \distcount{} algorithm in the $F$-turnstile model (Definition \ref{def:fturnstile}) for any finite field $F$, called $F$-\pcsa{}. We prove that $F$-\pcsa{} gives an unbiased estimate of $\norm{x}_{0;F}$ with $O(1/\sqrt{m})$ relative error (the hidden constant depends on $|F|$, see Table \ref{tab:psi} on page \pageref{tab:psi})  using $(m \log n \cdot \log |F|)$ bits of memory. Let $\widehat{\norm{x}_0}$ be the estimator of $\norm{x}_0$. We say it is \emph{unbiased} if for any $x\in F^n$,
\begin{align*}
    \E \widehat{\norm{x}_0} = \norm{x}_0.
\end{align*}
We say $\widehat{\norm{x}_0}$ has a \emph{relative error} $\epsilon$ if for any $x\in F^n$,
\begin{align*}
    \frac{\sqrt{\var \widehat{\norm{x}_0}}}{\norm{x}_0} = \epsilon.
\end{align*}

In 2023,   Bæk, Pagh and Walzer \cite{baek2023simple} analyzed a set sparse recovery algorithm over $\F_2$. They commented in the introduction that
``Linear sketches
over finite fields are less well-studied [in comparison to sketches over $\Z$ and $\R$], but are natural in some applications.'' \cite{baek2023simple}. We have a similar motivation here for $\F_2$-\pcsa{}. $\F_2$-\pcsa{} can be used to count the number of elements that appear an odd number of times, or the number of 1s in a boolean system. In general, for a prime $p$, $\F_p$-\pcsa{} can be used to count the number of elements whose frequencies are \emph{not} divisible by $p$.

Now consider another series of finite fields. For any positive integer $k$, $\F_{2^k}$-\pcsa{} can be used to estimate the number of non-zero $x_v$s where each $x_v$ is $\F_{2^k}$-valued (thus can be stored as a $k$-bit string). The state $x_v$ can be used to identify a vector of $k$ boolean attributes of the element $v$. Then given a stream of attribute toggles (in the form of ``toggle the $j$th attribute of the element $v$''), $\F_{2^k}$-\pcsa{} estimates the number of elements who have at least one nonzero attribute. 

When the prime $p$ is large, the $\F_p$-turnstile model approximates the original $\Z$-turnstile model. Thus as an application of the generic $F$-\pcsa{} algorithm, 
$\F_p$-\pcsa{} can be used to estimate $L_0$ when $p$ is propertly selected. By different strategies of selecting the prime $p$, it provides the following space complexity upper bounds (assuming a random oracle). To get a $(1\pm \epsilon)$-estimation of $L_0$ with 1/3 error probability, one needs
\begin{itemize}
    \item $O(\epsilon^{-2}\log n \cdot (\log \epsilon^{-1}+\log\log \norm{x}_\infty))$ bits, matching the state of the art $L_0$-estimation algorithm by Kane, Nelson and Woodruff \cite{kane2010optimal};
    \item $O(\epsilon^{-2}\log n \cdot \log \norm{x}_\infty)$ bits, which is useful when the maximum frequency is small;
    \item $O(\epsilon^{-2}\log n \cdot \max(\log \frac{\norm{x}_1}{\norm{x}_0}, \log \epsilon^{-1})) $ bits, which is useful when the average non-zero frequency is low.
\end{itemize}
Note that the upper bounds above can be made sharper by an \emph{$O(\epsilon)$-mass discount}, i.e.~ignoring the top $O(\epsilon)$-fraction of elements with non-zero frequencies. 
The key observation is: The worst thing that can happen in $\F_p$-\pcsa{} when $x_v$ is unusually large for some element $v\in U$ is that $v$ gets missed  in the final estimator (as if $v$ has never appeared in the stream).  It is fine to miss an $O(\epsilon)$-fraction of elements to get a $(1\pm \epsilon)$-estimation of $\norm{x}_0$, which is why in $\F_p$-\pcsa{} one may simply ignore the top $O(\epsilon)$-fraction of elements with high (even unbounded) frequencies.

Not all sketches enjoy this $O(\epsilon)$-mass discount feature. For example, the first $L_0$-estimation algorithm \cite{cormode2003comparing} uses a linear combination of $p$-stable random variables with a small $p$ (e.g.~$p<0.01$) to approximate $L_0$, which is sensitive to sparse high frequencies. To see this, let $x_1$ be the only non-zero entry in the vector $x$. For a fixed $p$, no matter how small $p$ is, the estimator $\widehat{\norm{x}_0}_{\text{$p$-stable}}$ becomes arbitrarily large as $x_1$ increases to infinity, while $\norm{x}_0$ is just $ 1$ by construction. 

\subsection{Notations}
Let $\ind{\cdot}$ be the indicator function. 

Let $U=\{1,2,\ldots,n\}$ and $x=(x_1,x_2,\ldots,x_n)\in \Z^n$. 
\begin{center}
    \begin{tabular}{c|c | c}
       notation  & definition & name \\
       \hline
      $\norm{x}_0$   &  $|\{v\in U\mid x_v\neq 0\}|$ & $L_0$, Hamming norm\\
      $\norm{x}_1$   &  $\sum_{v\in U} |x_v|$ & $L_1$\\
      $\norm{x}_\infty$ & $\max_{v\in U} |x_v|$ & maximum, $L_\infty$
    \end{tabular}
\end{center}

Let $F$ be any finite field. Let $|F|$ be the order of $F$. By default, 0 denotes the additive identity of $F$ and 1 denotes the multiplicative identity of $F$. The definition of $L_0$ can be extended to fields, i.e.~$\norm{x}_0=|\{v\in U\mid x_v\neq 0\}|$ where $x_v$s are $F$-valued and $0$ is the additive identity of $F$. 
When we want to emphasize that $x$ is $F$-valued, we write $L_0$ as $\norm{x}_{0;F}$. For any prime $p$ and positive integer $k$, we denote the field of order $p^k$ as $\F_{p^k}$.

Let $X$ be a real random variable. Let $\E X$ be its expectation and $\var X$ be its variance. Let $Y,Z$ be two generic\footnote{For example, random variables over a finite field.} random variables. We write $Y\sim Z$ if $Y$ and $Z$ are identically distributed.

Let $m$ be a positive integer. $[m]$ is a shorthand for the set $\{1,2,\ldots,m\}$.

\subsection{Preliminaries}
Commutative monoids are the most general algebraic structures to be updated in parallel. 
\begin{definition}[commutative monoids]
Let $M$ be a set and $+:M\times M\to M$ be a binary operator over $M$. Then $M$ is a \emph{commutative monoid} if
\begin{itemize}
    \item $\forall a,b,c\in M$, $(a+b)+c = a+(b+c)$; (associativity)
    \item $\forall a,b \in M$, $a+b = b+a$; (commutativity)
    \item $\exists 0 \in M\,\forall a\in M, a+0=0+a=a$. (existence of identity)
\end{itemize}
\end{definition}

Finite fields are the central objects in this work. Finite fields can be unambiguously identified by their orders up to isomorphism. Thus we write  the finite field of order $q$ as $\F_q$, where $q$ is necessarily $p^k$ for some prime $p$ and some positive integer $k$. 

\subsection{Overview of $F$-\pcsa{}}
\label{sec:overview}
Since $F$-\pcsa{} is similar to \pcsa{} \cite{flajolet1985probabilistic}, we will first recap the definition of \pcsa{}.
\subsubsection{\pcsa{} Sketch}
\begin{definition}[\pcsa{} \cite{flajolet1985probabilistic}]\label{def:pcsa_sa}
Let $U=\{1,2,\ldots,n\}$ be the universe.
The \pcsa{} sketch stores a table of bits, \bitmap$[i,j]$ where $i\in[m]$ and $j\in \N$. For every $(i,j)$, \bitmap$[i,j]$ is initialized to zero. There is a hash function $h:U\to [m]\times \N$ such that for any element $v\in U$, $\pr(h(v)=(i,j))=\frac{1}{m}2^{-j}$ for each $(i,j)$. Upon an insertion of $v\in U$ (or receiving $(v,1)$ in the $\N$-turnstile model), 
\begin{align*}
    \bitmap[h(v)] \gets \bitmap[h(v)] \lor 1.
\end{align*}    
\end{definition}
See Figure \ref{fig:pcsa} for an example of \pcsa{}'s memory snapshot. 

\begin{figure}
    \centering
    \begin{subfigure}[b]{0.3\textwidth}
         \centering
         \includegraphics[width=\textwidth]{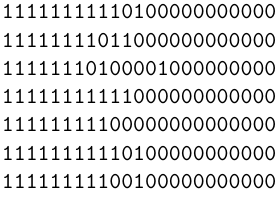}
         \caption{\pcsa{}'s memory}
         \label{fig:pcsa}
    \end{subfigure}
    \hfill
    \begin{subfigure}[b]{0.3\textwidth}      
         \centering
         \includegraphics[width=\textwidth]{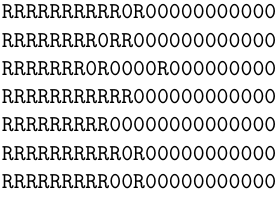}
         \caption{$F$-\pcsa{}'s memory}
         \label{fig:fpcsa}
    \end{subfigure}
    \hfill
    \begin{subfigure}[b]{0.3\textwidth}      
         \centering
         \includegraphics[width=\textwidth]{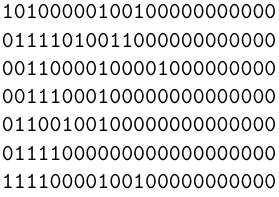}
         \caption{$\F_2$-\pcsa{}'s memory}
         \label{fig:f2pcsa}
    \end{subfigure}
    
    \caption{All three sketches are coupled using the same $h$ hash function. The way that ``\texttt{1}''s appear in \pcsa{}  is statistically the same with the way ``\texttt{R}''s appear in $F$-\pcsa{}. One does observe concrete ``\texttt{1}''s in \pcsa{}'s memory. On the other hand, ``\texttt{R}''s in $F$-\pcsa{} are observed in the memory as independent instances of uniform random varaibles over $F$. As an example, every bit that should have been a ``\texttt{1}'' in \pcsa{} now has a half chance to be a ``\texttt{0}'' in $\F_2$-\pcsa{}.}
    \label{fig:insight}
\end{figure}

Let $Q_i$ be the position of the leftmost zero in the $i$th row. It is observed in \cite{flajolet1985probabilistic} that whenever the cardinality doubles, $Q_i$ increases by 1 on average. The \pcsa{}'s estimator is thus
\begin{align*}    \widehat{\norm{x}_0}_{\text{\pcsa{}}} = \alpha_m 2^{\frac{1}{m}\sum_{i=1}^m Q_i},
\end{align*}
where $\alpha_m$ is some normalization factor.
\paragraph{Middle Range Assumption.}
We assume the cardinality is in the middle range in this work, i.e.~$\norm{x}_0\gg 1$ and $\norm{x}_0 \ll |U|$. In this case, all bits $(i,j)$ with $j\geq \log |U|$ are zeros with high probability. Thus we only need to store bits $(i,j)$ with $j\in [0,\log |U|]$, which requires $(m \log |U|)$ bits of memory. For the case where $\norm{x}_0=\Theta(1)$ or $\norm{x}_0= \Theta(|U|)$, separate estimation methods are needed, which are out of the scope of this paper.

\subsubsection{$F$-\pcsa{} Sketch}\label{sec:fpcsa_overview}
The algebraic structure of finite fields brings us the following properties.
\begin{lemma}\label{lem:idem}
 Let $F$ be a finite field. Then the uniform distribution over $F$ is well-defined.
 
 Let $R_1,R_2,\ldots,R_k$ be i.i.d.~uniform random variables over $F$. Then for any non-zero $a\in F$ and $k\in \Z_+$,
 \begin{align}
     a R_1 \sim R_1\label{eq:field_idem}\\
     \sum_{i=1}^k R_i \sim R_1.\label{eq:idem_measure}
 \end{align}
\end{lemma}
\begin{proof}
The uniform distribution over a finite field $F$ is just the normalized counting measure,  which is a well-defined probability measure since $F$ is finite.

    Since $F$ is a field and $a$ is non-zero, there is a multiplicative inverse $a^{-1}$ of $a$. Then we have, for any $b\in F$,
    \begin{align*}
        \pr(a R_1 = b) = \pr(R_1 = ba^{-1}) = \frac{1}{|F|}.
    \end{align*}
    Thus (\ref{eq:field_idem}) is proved. For (\ref{eq:idem_measure}), it suffices to prove $R_1+R_2\sim R_1$ and then use induction. For any $b\in F$,
    \begin{align*}
        \pr(R_1+R_2 = b) &=\sum_{c\in F}\frac{1}{|F|}\pr(R_1 + c= b)
        \intertext{since every element $c$ has an additive inverse $-c$ in $F$,}
        &=\sum_{c\in F}\frac{1}{|F|}\pr(R_1 = b-c) = \sum_{c\in F}\frac{1}{|F|^2} =\frac{1}{|F|},
    \end{align*}    
\end{proof}

The lemma above enables us to count over $F$ in an almost identical way as \pcsa{} counts over $\N$.
\begin{definition}[$F$-\pcsa{}]\label{def:fpcsa_sa}    
Let $U=\{1,2,\ldots,n\}$ be the universe.
The $F$-\pcsa{} sketch stores a table of $F$-values, \fieldmap[$i$,$j$] where $i\in[m]$ and $j\in \N$. For every $(i,j)$, \fieldmap$[i,j]$ is initialized to zero (the additive identity of $F$). There is a hash function $h:U\to [m]\times \N$ such that for any $v\in U$, $\pr(h(v)=(i,j))=\frac{1}{m}2^{-j}$ for each $(i,j)$. There is another hash function $g:U\to F$ such that for any $v\in U$, $g(v)$ is uniformly random over $F$. Upon receiving $(v,k)$ where $v\in U,k\in F$, 
\begin{align*}
    \fieldmap[h(v)] \gets \fieldmap[h(v)] + k\cdot g(v).
\end{align*}    
\end{definition}
Together with Lemma \ref{lem:idem}, it is immediate that, conditioning on the hash function $h$, the entry $\fieldmap[i,j]$ behaves as follows. 
\begin{align*}
    \fieldmap[i,j] &= 0, \quad \text{if $x_v=0$ for all $v\in U$ that $h(v)=(i,j)$}\\
    \fieldmap[i,j] &\sim R, \quad \text{if $x_v\neq 0$ for some $v\in U$ that $h(v)=(i,j)$},
\end{align*}
where $R$ is a uniform random variable over $F$. See Figure \ref{fig:fpcsa} for a diagram of $F$-\pcsa{}'s memory. $F$-\pcsa{} behaves statistically the same with the original \pcsa{} with the twist that ``1''s becomes ``\texttt{R}''s. This is the intuition why it is no harder to count in the $F$-turnstile model than in the $\N$-turnstile model. One needs $\log |F|$ bits per register to store $F$-values and thus $F$-\pcsa{} needs $(m\log |U|\cdot \log |F|)$ bits of memory.

We use the \emph{highest index} (the index of the rightmost non-zero value) estimator for $F$-\pcsa{}. Let $W_i$ be the highest index of non-zero entries in the $i$th row. Then the estimator is defined as
\begin{align*}
\widehat{\norm{x}_{0;F}}_{\text{$F$-\pcsa{}}} = \beta_m 2^{\frac{1}{m}\sum_{i=1}^m W_i},
\end{align*}
where $\beta_m$ is some normalization factor. We prove, under some simplification assumptions listed in Section \ref{sec:assumptions}, in Theorem \ref{thm:fpcsa}  that $\widehat{\norm{x}_{0;F}}_{\text{$F$-\pcsa{}}}$ is unbiased with $O(1/\sqrt{m})$ relative error (the leading constant depends on $|F|$). See Table \ref{tab:psi} for computed normalization factors and relative errors for a list of finite fields.

\begin{table}
    \centering
    \begin{tabular}{c|c|c| c|c }
        $|F|$ & $\psi_\E(|F|)$ & $\psi_\var(|F|)$ &  normalization factor ($\beta_m$)  & relative error  \\
        \hline
2 & -0.1100   &     5.588      & 1.079$m$  &  1.638$/\sqrt{m}$  \\
3 & 0.5487 & 4.321             & 0.6836$m$    &  1.441$/\sqrt{m}$  \\
4 & 0.7903 & 4.002             & 0.5782$m$   &  1.387$/\sqrt{m}$  \\
5 & 0.9173 & 3.861             & 0.5295$m$    &  1.362$/\sqrt{m}$  \\
6 & not a field & not a field  &  not a field   &  not a field  \\
7 & 1.049 & 3.732              & 0.4833$m$    &  1.339$/\sqrt{m}$  \\
8 & 1.088 & 3.697              & 0.4704$m$  & 1.333$/\sqrt{m}$ \\
9 & 1.118 & 3.671              &  0.4607$m$ & 1.328$/\sqrt{m}$ \\
10 & not a field & not a field & not a field  & not a field\\
$2^8$ & 1.326 & 3.512          & 0.3989$m$  & 1.299$/\sqrt{m}$ \\
$2^{32}$ & 1.333 & 3.507       & 0.3969$m$  & 1.298$/\sqrt{m}$\\          
$\infty$ & 1.333 & 3.507       & 0.3969$m$  & 1.298$/\sqrt{m}$
    \end{tabular}
    \caption{Asymptotic constants as $m\to \infty$. By Theorem \ref{thm:fpcsa}, the normalization factor is $2^{-\psi_\E (|F|)}m$ and the relative error is $\sqrt{m^{-1}(\log^2 2)\psi_\var(|F|)}$, where $\psi_\E(|F|)$ and $\psi_\var(|F|)$ are defined in Theorem \ref{thm:psi_asym} (Appendix \ref{app:math}). The values of   $\psi_\E(|F|)$ and $\psi_\var(|F|)$ in the table are obtained through numerical integration. Interestingly,
    as $|F|$ goes to infinity, $F$-\pcsa{} estimator statistically converges to the \textsf{LogLog} estimator \cite{durand2003loglog} and  the relative error  computed here (1.298$/\sqrt{m}$)  matches Durand and Flajolet's computation for the \textsf{LogLog}'s relative error,  which is $m^{-1/2}\sqrt{\frac{1}{12}\log^2 2+\frac{1}{6}\pi^2}\approx 1.29806/\sqrt{m}$.}
    \label{tab:psi}
\end{table}

\subsection{Related Work}
To the best of our knowledge, the problem of \distcount{} in the $F$-turnstile model where $F$ is a generic finite field is new. Thus $\F_{2^k}$-\pcsa{} are interesting new counting algorithms in $k$-boolean systems. $\F_p$-\pcsa{} where $p$ is a small prime is also new to count the number of elements whose frequency is not divisible by $p$.

When $p$ is large, $\F_p$-turnstile goes to $\Z$-turnstile and the \distcount{} problem in the $\F_p$-turnstile model is intimately related to the $L_0$-estimation in the $\Z$-turnstile model. Indeed, finite fields $\F_p$ when $p$ is large are used in the construction of the $L_0$-estimation algorithm proposed by
Kane, Nelson and Woodruff \cite{kane2010optimal} (\knw-$L_0$). They randomly project the vector hashed to each cell to an element in $\F_p$, with some random large prime $p$. When $p$ is large, $\F_p$-\pcsa{} and \knw-$L_0$ store statistically very similar information (a multi-resolution table of  random projections over $\F_p$). However, the estimation methods are different. \knw-$L_0$ first uses a subroutine \textsc{RoughEstimator} to roughly estimates $\norm{x}_0$. From this rough estimate, \knw-$L_0$ chooses one column (see Figure \ref{fig:fpcsa}) with a constant fraction of ``\texttt{R}''s. Then they prove that by choosing a large $p$ randomly, no ``\texttt{R}'' will be a zero \emph{in that column}. Conditioning on this event, the $L_0$-estimation can be solved by their $F_0$-estimation algorithm (in the incremental setting), through balls-and-bins statistics in that column. On the other hand, the estimator of $\F_p$-\pcsa{} inherently takes account of the events 
``$\texttt{R}=\texttt{0}$'' and uses the highest index to give an estimator with exact and sharper error constants. The difference between the two estimation methods becomes more significant when $p$ is small, e.g.~$p=2$, where half of the ``\texttt{1}''s will be ``$\texttt{0}$''s in the table (see Figure \ref{fig:f2pcsa}) and there is no way to condition with high probability on the event that ``$\texttt{R}\neq \texttt{0}$'' for a whole column.

\subsection{Organization}
$F$-\pcsa{} is formally analyzed in
Section \ref{sec:fpcsa}. Applications of $F$-\pcsa{} on $L_0$-estimation are discussed in Section \ref{sec:application}.

\section{Analysis of $F$-\pcsa{}\label{sec:fpcsa}}

\subsection{Simplification Assumptions}\label{sec:assumptions}
\begin{assumption}[random oracle]
    There is a random oracle that provides ideal hash functions $h$ and $g$.
\end{assumption}

\begin{assumption}[poissonization]
Since we only care about the cases where $\norm{x}_0$ is in the middle range (see the Middle Range Assumption in Section \ref{sec:overview}), we may assume the following.
\begin{itemize}
    \item Since every element is hashed to only one table entry (through the hash function $h(\cdot)$), there are negligible negative correlations between table entries. For example, if there is only one non-zero element, then a non-zero entry necessarily implies that all other entries are zeros. However, this correlation vanishes when $\norm{x}_0\gg 1$  and thus we will assume all the entries in \fieldmap{} are \emph{strictly} independent. I.e.~when a new element comes in, instead of being hashed to only one entry, now every entry $(i,j)$ has an independent chance of $\frac{1}{m}2^{-j}$ to be hashed to. 
    \item The probability that one hits the $(i,j)$ entry is $\frac{1}{m}2^{-j}$ which goes to $1-e^{-\frac{1}{m}2^{-j}}$ as $j$ becomes large. Since $\norm{x}_0\gg 1$ and the highest index is likely to be large, we will assume the probability of hitting the $(i,j)$ entry is \emph{exactly} $1-e^{-\frac{1}{m}2^{-j}}$. 
    \item Since $\norm{x}_0$ is assumed to be in the middle range, we assume that the columns are indexed by \emph{integers}, extending infinitely to both ends. E.g.~the probability of hitting the $(i,-3)$ entry is now assumed to be $1-e^{-\frac{2^3}{m}}$.
\end{itemize}
\end{assumption}
\begin{remark}
This assumption is called ``poissonization'' since it essentially assumes that elements are inserted \emph{continuously} into the sketch with rate 1, which is why the entries are independent and the $(i,j)$ entry has a hitting probability $1-e^{-\frac{1}{m}2^{-j}}$. Poissonization doesn't affect the statistics of high-positioned entries when $\norm{x}_0$ becomes mildly large, which can be shown by a standard coupling method. 
\end{remark}

There is another difficulty during the analysis of such estimators: The behavior of the estimator has constant fluctuations which do \emph{not} vanish even when $m$ goes to infinity. See,  for example, the analysis of \pcsa{} \cite{flajolet1985probabilistic} and \textsf{HyperLogLog} \cite{flajolet2007hyperloglog} where every statement has a periodic error term. See Figure 1 of \cite{pettie2021information} for a visualization of the fluctuation. As proposed in \cite{pettie2021information}, these fluctuations can be eliminated elegantly by  \emph{randomly offsetting} the subsketches.
\begin{assumption}[random offsetting]
    We assume that for each $i\in [m]$, there is a random offset $\Theta_i$ uniformly drawn from $[0,1)$ for the $i$th subsketch. Then the $i$th subsketch randomly ignores a $2^{-\Theta_{i}}$-fraction of the elements in the universe. Thus now the probability of hitting the $(i,j)$ entry is $1-e^{-m^{-1}2^{-j-\Theta_i}}$.
\end{assumption}
\begin{remark}
A \pcsa{}-style sketch is inherently multiplicatively periodic in $2$ and therefore one will get the same relative behaviour for $\norm{x}_0 = 2^k$ with different $k\in \N$. However, the sketch will have different behaviours for e.g.~$\norm{x}_0 = 2^k$ and $\norm{x}_0 = 2^{k+0.5}$. The random offsetting assumption adds noise to this process, thus making the sketch now has an ``averaged'' behaviour for any $\norm{x}_0\in \R^+$. Note that this random offsetting trick is \emph{necessary to be actually implemented} if the required error is comparable to the fluctuation magnitudes (usually small, e.g.~$<10^{-6}$ for \pcsa{} \cite{flajolet1985probabilistic}).
\end{remark}

With all the assumptions above, we reconstruct   the $F$-\pcsa{} sketch (Definition \ref{def:fpcsa_sa}) in the following subsection. 
\subsection{Construction of $F$-\pcsa{}}
With the poissonization assumptions, now $F$-\pcsa{} consists of $m$ i.i.d.~subsketches (rows in \fieldmap{}). We will construct one subsketch now.

Let $F$ be any finite field. Let the \emph{memory state} $s=(s_i)_{-\infty}^\infty$ be a bilaterally infinite vector over $F$ (corresponding to one row in \fieldmap). We will refer the state vector $x$ in the $F$-turnstile model (Definition \ref{def:fturnstile}) as the \emph{stream state}. Recall that the stream state is an $n$-vector over $F$ where $n$ is the size of the universe $U$. 

We first draw a random offset $\Theta\in [0,1)$.
Then with the (poissonized) hash functions $h$ and $g$ provided by the random oracle, we generate two tables of random variables $H=(H_{v,i})_{v\in U,i\in\Z}$ and $G=(G_{v,i})_{v\in U,i\in\Z}$, where for each $v\in U$ and $i\in \Z$,
\begin{itemize}
    \item $H_{v,i}$ is a $\{0,1\}$-random variable\footnote{ 0 is the additive identity in $F$; 1 is the multiplicative identity in $F$.} with $\pr(H_{v,i}=0) = e^{-m^{-1}2^{-i-\Theta}}$,
    \item $G_{v,i}$ is a uniform random variable over $F$. 
\end{itemize}
The entries in $H$ and $G$ are all independent. The two tables $H$ and $G$ are independent.

The sketching process can be described as follows. 
\begin{itemize}
    \item Initially, both the memory state $s$ and the stream state $x$ are all zeros.
    \item Upon receiving $(v,k)$ for $v\in U$ and $k\in F$,
    \begin{itemize}
        \item $x_v \gets x_v + k$ and
        \item for each $i\in \Z$
        \begin{itemize}
            \item $s_i \gets s_i + k\cdot H_{v,i} G_{v,i}$. 
        \end{itemize}
    \end{itemize}
\end{itemize}

Note that the process described above faithfully simulates the $F$-\pcsa{} in Definition \ref{def:fpcsa_sa} with the assumptions listed in Section \ref{sec:assumptions}.

We then formalize the intuition in Section \ref{sec:fpcsa_overview}.
\subsection{Characterization of $F$-\pcsa{}'s States}

\begin{theorem}[characterization of $F$-\pcsa{}'s state]\label{thm:char_typcsa_state}
Fix a stream state $x$ and an offset $\theta\in [0,1)$. Let the memory state be $S=(S_i)_{i\in \Z}$.\footnote{We capitalize the memory state $s$ when we emphasize it as a \emph{random vector} for a fixed stream state $x$.} Then
\begin{itemize}
    \item $S_i$s are mutually independent.
    \item $\pr(S_i\neq 0) = (1-e^{-\norm{x}_{0}m^{-1} e^{-i-\theta}}) (1-|F|^{-1})$, for $i\in \Z$.
\end{itemize}
\end{theorem}
\begin{proof}
By construction, the memory state $S$ and the stream state $x$ are linearly related:
\begin{align}
    S_i = \sum_{v\in U} x_v H_{v,i}G_{v,i}.\label{eq:six}
\end{align} Consequently, $S_i$s are independent due to the independence of entries of $H$ and $G$.

Recall that $G_{v,i}$s are i.i.d.~uniform random variables over $F$. By Lemma \ref{lem:idem}, we have
\begin{align*}
    S_i &= \sum_{v\in U} x_v H_{v,i}G_{v,i} = \sum_{v\in U:  x_v H_{v,i}\neq 0}  x_v H_{v,i}G_{v,i}\\
    &\sim \ind{\exists v, x_v H_{v,i}\neq 0} R_i,
\end{align*}
where $R_i$ is uniformly random over $F$. Thus we know
\begin{align}
    \pr(S_i\neq 0)  &= \pr (\exists v, x_v H_{v,i}\neq 0)\pr(R_i\neq 0)\nonumber\\
    &=(1-\pr (\forall v, x_v H_{v,i}= 0))(1-|F|^{-1})\nonumber\\
    &=(1-\prod_{v\in U}\pr(x_v H_{v,i}= 0))(1-|F|^{-1})\label{eq:s_inter}
\end{align}

Recall that $H_{v,i}$ is a binary random variable with probability $e^{-2^{-i/m-\theta}}$ of being 0. Since $F$ is a field, $ x_v H_{v,i}=0$ if and only if $(H_{v,i}=0 \lor x_v = 0)$. Thus
\begin{align*}
    \pr(x_vH_{v,i}=0) = \begin{cases}
        1, &x_v = 0\\
       e^{-m^{-1}2^{-i-\theta}}, &x_v \neq 0
    \end{cases},
\end{align*}
and then by (\ref{eq:s_inter}),
\begin{align}
    \pr(S_i\neq 0) &= (1-e^{-\norm{x}_0m^{-1} 2^{-i-\theta}}) (1-|F|^{-1}).\nonumber
\end{align}

\end{proof}

With the characterization theorem, we now analyze the estimator of $F$-\pcsa{}.
\subsection{Estimation of $F$-\pcsa{}}
The estimation is based on the highest non-zero position. The analysis is overall straightforward except for a technical asymptotic analysis included in the appendix. 
\begin{definition}[highest index/position]
Let $S$ be the memory state and $\Theta$ be the random offset. Then the \emph{highest index} $B$ is defined as
\begin{align*}
    B = \sup\{i\in \Z\mid S_i\neq 0\}
\end{align*}
and the \emph{highest position}\footnote{The highest position is just the smoothed version of the highest index.} is defined as
\begin{align*}
    W = B + \Theta.
\end{align*}
\end{definition}
\begin{definition}\label{def:nu}
Let $\nu:\R \times [0,1]\to \R$ be a function. For $z\in \R$ and $r\in [0,1]$,
    \begin{align*}
    \nu(z,r) &= (1-e^{- 2^{-z}}) r  \prod_{j=1}^\infty \left(1- (1-e^{- 2^{-z-j}}) r\right) 
    \end{align*}
\end{definition}
\begin{lemma}[p.d.f.~of $W$]\label{lem:density}
    Let $f_W$ be the probability density function of $W$. Then
    \begin{align}
        f_W(z) &= \nu(z-\log_2(\norm{x}_0m^{-1}),1-|F|^{-1}). \label{eq:density}
    \end{align}
\end{lemma}
\begin{proof}    
For $i\in \Z$, by the characterization theorem (Theorem \ref{thm:char_typcsa_state}),
\begin{align*}
    \pr( B \leq i |\Theta =\theta) = \prod_{j=i+1}^\infty \pr(S_j= 0)= \prod_{j=i+1}^\infty \left(1- (1-e^{-\norm{x}_0m^{-1} 2^{-j-\theta}}) (1-|F|^{-1})\right)
\end{align*}
Thus
\begin{align}
    &\pr( B = i |\Theta =\theta)\nonumber\\
    &= \pr( B \leq i+1 |\Theta =\theta) - \pr( B \leq i |\Theta =\theta)\nonumber\\
    &=(1-e^{-\norm{x}_0m^{-1} 2^{-i-\theta}}) (1-|F|^{-1}) \prod_{j=i+1}^\infty \left(1- (1-e^{-\norm{x}_0m^{-1} 2^{-j-\theta}}) (1-|F|^{-1})\right) \nonumber  \\
    &= \nu(i+\theta-\log_2(\norm{x}_0m^{-1}),  1-|F|^{-1}) \label{eq:b=i}.
\end{align}
Then we have
\begin{align*}
    \int_{-\infty}^{z_0}f_W(z)dz &= \pr(W\leq z_0) = \pr(B+\Theta\leq z_0) = \int_0^1 \pr(B+\theta\leq z_0|\Theta=\theta)\,d\theta\\
    &= \int_0^1 \sum_{i=-\infty}^\infty \ind{i+\theta \leq z_0}\pr( B = i |\Theta =\theta) \,d\theta
    \intertext{by (\ref{eq:b=i}),}
    &=\int_0^1 \sum_{i=-\infty}^\infty \ind{i+\theta \leq z_0}\nu(i+\theta-\log(\norm{x}_0m^{-1}),  1-|F|^{-1}) \,d\theta\\
    &=\int_{-\infty}^\infty \ind{z \leq z_0}\nu(z-\log_2(\norm{x}_0m^{-1}),  1-|F|^{-1}) \,dz\\
    &=\int_{-\infty}^{z_0} \nu(z-\log_2(\norm{x}_0m^{-1}),  1-|F|^{-1}) \,dz.
\end{align*}
Differentiate both sides by $z_0$ and we get (\ref{eq:density}).
\end{proof}

\begin{definition}\label{def:phi}
    Let $\phi:(0,1)\times [0,1]\to \R$ be a function. For $t\in(0,1)$ and $r\in [0,1]$,
    \begin{align*}
        \phi(t,r) = \int_{-\infty}^\infty 2^{tz}  \nu(z,r)   \,dz.
    \end{align*}
\end{definition}

\begin{lemma}[base-2 moment generating function of $W$]\label{lem:mgf}
For $t\in(0,1)$ and $|F|\geq 2$,
    \begin{align*}
        \E 2^{tW} =\norm{x}_0^t m^{-t}\phi(t,1-|F|^{-1}).
    \end{align*}
\end{lemma}
\begin{proof}
By Lemma \ref{lem:density},
\begin{align*}
    \E 2^{tW} = \int_{-\infty}^\infty f_W(z)2^{tz}\,dz &= \int_{-\infty}^\infty \nu(z-\log_2(\norm{x}_0m^{-1}),1-|F|^{-1})2^{tz}\,dz
    \intertext{set $s = z-\log_2(\norm{x}_0 m^{-1})$}
    &=\norm{x}_0^tm^{-t}\int_{-\infty}^\infty \nu(s,1-|F|^{-1})2^{ts}\,ds
    \intertext{by the definition of $\phi(\cdot,\cdot)$ (Definition \ref{def:phi})}
    &=\norm{x}_0^tm^{-t} \phi(t,1-|F|^{-1}).
\end{align*}

\end{proof}

\begin{theorem}[$F$-\pcsa{}]\label{thm:fpcsa}
    Let $F$ be any finite field and $m\geq 3$ be the number of i.i.d.~sub\-sketches. Let $W_1,W_2,\ldots,W_m$ be the highest positions of the subsketches. Then the $F$-\pcsa{} estimator is defined as
    \begin{align*}
        \widehat{\norm{x}_0} =\phi(1/m, 1-|F|^{-1})^{-m}m 2^{\frac{1}{m}\sum_{i=1}^m W_i}.
    \end{align*}
    The estimator is unbiased,
    \begin{align*}
        \E \widehat{\norm{x}_0} = \norm{x}_0,
    \end{align*}    
    and has variance
    \begin{align*}
        \var \widehat{\norm{x}_0} &=  (\phi(1/m, 1-|F|^{-1})^{-2m} \phi(2/m, 1-|F|^{-1})^{m}-1) \norm{x}_0^2.
    \end{align*}
    When $m$ is large,
    \begin{align*}
    \widehat{\norm{x}_0} &= (2^{-\psi_\E(|F|)}+O(m^{-1}))m 2^{\frac{1}{m}\sum_{i=1}^m W_i}\\
        \var \widehat{\norm{x}_0} &= \norm{x}_0^2 m^{-1} \left((\log^2 2)\psi_\var(|F|)+ O(m^{-1})\right),
    \end{align*}
    where $\psi_\E$ and $\psi_\var$ are defined in Theorem \ref{thm:psi_asym} in Appendix \ref{app:math}. The relative error is bounded
    \begin{align}
        \sqrt{(\log^2 2)\psi_\var(|F|)} < \sqrt{79},\label{eq:rough_bound}
    \end{align}
    for all finite fields $F$. 
\end{theorem}
\begin{remark}
    The rough bound (\ref{eq:rough_bound}) is only to show the relative error is $O(1)$ regardless of $F$. The actual relative errors are much smaller. 
    See Table \ref{tab:psi} on page \pageref{tab:psi} for some computed values of relative errors.
\end{remark}
\begin{proof}
Since $W_i$s are independent, we have
    \begin{align*}
        \E \widehat{\norm{x}_0} &= \phi(1/m, 1-|F|^{-1})^{-m} m\E 2^{\frac{1}{m}\sum_{i=1}^m W_i}\\
        &=\phi(1/m, 1-|F|^{-1})^{-m} m\left(\E 2^{\frac{1}{m} W_1}\right)^m
        \intertext{by Lemma \ref{lem:mgf},}
        &=\phi(1/m, 1-|F|^{-1})^{-m} m\left(\norm{x}_0^{1/m}m^{-1/m}\phi(1/m,1-|F|^{-1})\right)^m\\
        &= \norm{x}_0.
    \end{align*}
Thus $ \widehat{\norm{x}_0}$ is unbiased.
Similarly, 
\begin{align*}
    \E \widehat{\norm{x}_0}^2 &=\phi(1/m, 1-|F|^{-1})^{-2m}m^2 \left(\E 2^{\frac{2}{m} W_1}\right)^{m} \\
    &= \phi(1/m, 1-|F|^{-1})^{-2m} m^2  \left( \norm{x}_0^{2/m}m^{-2/m}\phi(2/m,1-|F|^{-1})\right)^{m}\\
    &= \norm{x}_0^{2}\phi(1/m, 1-|F|^{-1})^{-2m}   \phi(2/m,1-|F|^{-1})^{m}.
\end{align*}
The variance is obtained by the formula $\var \widehat{\norm{x}_0} = \E \widehat{\norm{x}_0}^2 - (\E \widehat{\norm{x}_0} )^2$.

The asymptotic case when $m$ is large is proved in Theorem \ref{thm:psi_asym} in Appendix \ref{app:math}.

\end{proof}

\section{Applications to $L_0$-Estimation}\label{sec:application}
First note that an unbiased algorithm with relative error $O(\epsilon)$ implies a $(1\pm \epsilon)$-estimation algorithm with constant error probability by Chebyshev's inequality. Thus $F$-\pcsa{} gives an $(1\pm \epsilon)$-estimation of $\norm{x}_{0;F}$ with a constant error probability in the $F$-turnstile model using $O(\epsilon^{-2}\log n \cdot \log |F|)$ bits of memory.
\subsection{A Matching Upper Bound with Kane, Nelson and Woodruff's $L_0$-Estimation Algorithm}
We will prove that it is possible to achieve the same space complexity of \knw-$L_0$ (\cite{kane2010optimal}) with $\F_p$-\pcsa{} using a similar $p$-selection trick. Instead of bounding the probabilities of  all the ``bad events'' as in \cite{kane2010optimal}, for $\F_p$-\pcsa{} we only need to estimate the difference between $\norm{x}_{0}$ and $\norm{x}_{0;\F_p}$, where $\norm{x}_{0}$ is the quantity we \emph{want} to estimate and $\norm{x}_{0;\F_p}$ is the quantity we \emph{can} estimate accurately by $\F_p$-\pcsa{}.
\begin{theorem}\label{thm:random_p}
Let $\epsilon>0$ be the desired error. Let $P$ be a random prime uniformly drawn from the set $S$ of the smallest $\frac{10\log \norm{x}_\infty}{\epsilon}$ primes. Then for any $x\in \Z^n$
\begin{align*}
\norm{x}_{0} \geq \norm{x}_{0;\F_P} , \quad \text{and}\quad 
    \pr(\norm{x}_{0}-\norm{x}_{0;\F_P}>\epsilon \norm{x}_0) \leq \frac{1}{10}.
\end{align*}
\end{theorem}
\begin{proof}
    Since for any $a\in \Z$, $a=0$ implies $a=0\mod P$, we have $\norm{x}_{0} \geq \norm{x}_{0;\F_P}$. 

 Imagine every element is placing balls in a set of bins indexed by $S$. Each element $v\in U$ will place a ball in the bin indexed by $p\in S$ if $p\mid x_v$. By randomly picking a prime $P$ from the set $S$, $\norm{x}_0-\norm{x}_{0;\F_P}$ is equal to the number of element $v$ with $P\mid x_v$, which is equal to the number of balls in the bin indexed by $P$. 
    Note that for every element $v\in U$, $x_v$ has at most $\log \norm{x}_\infty$ distinct prime  factors. Thus the number of balls is at most $\norm{x}_0\log \norm{x}_\infty$. As a result, the number of bins with at least $\epsilon\norm{x_0}$ balls is at most $\log \norm{x}_\infty/\epsilon$ (otherwise we run out of balls). We conclude that
    \begin{align*}
        \pr(\norm{x}_{0}-\norm{x}_{0;\F_P}>\epsilon \norm{x}_0) &\leq \frac{\text{number of bins with at least $\epsilon \norm{x}_0$ balls}}{\text{number of bins}}\\
        &\leq \frac{\log \norm{x}_\infty/\epsilon}{{10\log \norm{x}_\infty}/{\epsilon}}=1/10.
    \end{align*}
\end{proof}
\begin{remark}
    In \knw-$L_0$, a random prime is picked in the range $[O(\frac{\log \norm{x}_\infty}{\epsilon}),O((\frac{\log \norm{x}_\infty}{\epsilon})^3)]$. In contrast, we only need to pick a random prime in the range $[2,\tilde{O}(\frac{\log \norm{x}_\infty}{\epsilon})]$ for $\F_p$-\pcsa{}.
\end{remark}
To store a prime in Theorem \ref{thm:random_p}, one needs $(\log \epsilon^{-1}+\log\log \norm{x}_\infty + O(1))$ bits of memory. Therefore, by choosing a random $p$ with the scheme in Theorem \ref{thm:random_p}, $\F_p$-\pcsa{} gets the same space complexity $O(\epsilon^{-2}\log n(\log \epsilon^{-1}+\log\log \norm{x}_\infty))$ with \knw-$L_0$, to give a $(1\pm \epsilon)$-estimation of $L_0$ with a constant error probability, but with a much smaller hidden constant.

\subsection{New Space Complexity Upper Bounds for $L_0$-Estimation}
\subsubsection{$O(\epsilon^{-2}\log n\cdot \log \norm{x}_\infty )$}
Let $p$ be a prime number.
Trivially, if $\norm{x}_\infty < p$ then $\norm{x}_{0}=\norm{x}_{0;\F_p}$. Therefore by choosing $p$ of size $\Tilde{O}(\norm{x}_\infty)$, $\norm{x}_{0}$ is perfectly estimated by $\F_p$-\pcsa{}, which uses $O(\epsilon^{-2}\log n \log \norm{x}_\infty )$ space. Most notably, if the system is boolean ($x_v$s toggle between 0 and 1) or the system's insertions and deletions are interleaved, then $\norm{x}_\infty<2$. In such case, $\F_2$-\pcsa{} estimates $\norm{x}_0$ perfectly with only $O(\epsilon^{-2}\log n)$ bits of memory.

\subsubsection{$O(\epsilon^{-2}\log n \cdot \max (\log(\frac{\norm{x}_1}{\norm{x}_0}),  \log(\epsilon^{-1})))$}
Observe the following.

\begin{theorem}\label{thm:avg_p}
Let $\epsilon>0$ be the desired error.
Let $S=\{p_1,p_2,\ldots,p_K\}$ be a set of $K$ distinct primes such that $\min S \geq \norm{x}_1/\norm{x}_0$ and $K\geq 10/\epsilon$. Let $P$ be a random prime uniformly drawn from S. Then for any $x\in \Z^n$,
\begin{align*}
\norm{x}_{0} \geq \norm{x}_{0;\F_P} , \quad \text{and}\quad 
    \pr(\norm{x}_{0}-\norm{x}_{0;\F_P}>\epsilon \norm{x}_0) \leq \frac{1}{10}.
\end{align*}
\end{theorem}
\begin{proof}
We only need to prove the second inequality.
Let $y_k=|\{v: p_k | x_v,x_v\neq 0\}|$, i.e.~the number of non-zero $x_v$s that have prime factor $p_k$. Without loss of generality, assume $x_1,x_2,\ldots,x_\tau$ are positive and $x_{\tau+1},\ldots,x_n$ are zeros. Then by the AM-GM inequality, 
\begin{align*}
    \frac{\sum_{i=1}^\tau x_i}{\tau} \geq \left(\prod_{i=1}^\tau x_i\right)^{\tau^{-1}} \geq \left(\prod_{i=1}^K p_i^{y_i}\right)^{\tau^{-1}} \geq p_*^{\tau^{-1}\sum_{i=1}^K y_i},
\end{align*}
where $p_*=\min S$. Thus we have $\sum_{i=1}^K y_i \leq \tau \frac{\log(\norm{x}_1/\tau )}{\log p_*}$.
Note that
\begin{align*}
    \pr(\norm{x}_{0}-\norm{x}_{0;\F_P}\geq \epsilon \norm{x}_0)&=\sum_{i=1}^K \frac{1}{K}\ind{y_k\geq \epsilon \norm{x}_0}
    \intertext{where we use $\ind{y_k\geq \epsilon \norm{x}_0} \leq \frac{y_k}{\epsilon \norm{x}_0}$,}
    &\leq \sum_{i=1}^K \frac{1}{K\epsilon \norm{x}_0}y_k= \frac{1}{K\epsilon \norm{x}_0}\sum_{i=1}^K y_k.
    \intertext{since $\sum_{i=1}^K y_i \leq \tau \frac{\log(\norm{x}_1/\tau )}{\log p_*}$}
&\leq \tau \frac{\log(\norm{x}_1/\tau)}{\log p_*} \frac{1}{K\epsilon \norm{x}_0} 
 \intertext{ by definition, $\tau = \norm{x}_0$}
 &= \frac{\log(\norm{x}_1/\norm{x}_0)}{\log p_*} \frac{1}{K\epsilon} 
\end{align*}
By assumption, $p_*=\min S \geq \norm{x}_1/\norm{x}_0$ and $K\geq 10/\epsilon$. We conclude that $ \pr(\norm{x}_{0}-\norm{x}_{0;\F_P}\geq \epsilon \norm{x}_0) \leq 1/10$.
\end{proof}
To store a prime in Theorem \ref{thm:avg_p}, one needs $\log (\Tilde{O}(\max(\norm{x}_1/\norm{x}_0, \epsilon^{-1} )))$ bits of space.   Therefore, by choosing a random $p$ with the scheme in Theorem \ref{thm:avg_p}, $\F_p$-\pcsa{} gets a new space complexity $O(\epsilon^{-2}\log n \cdot \max (\log(\frac{\norm{x}_1}{\norm{x}_0}),  \log(\epsilon^{-1})))$ to give a $(1\pm \epsilon)$-estimation of $L_0$ with a constant error probability.

Note that $\norm{x}_1/\norm{x}_0$ is precisely the \emph{average non-zero frequency}. Thus this scheme is particularly space-saving in a system where the maximum frequency is unbounded but the average frequency is small. 
\begin{remark}
    Similar upper bounds can be proved in the same way with $\frac{\norm{x}_1}{\norm{x}_0}$ replaced by $( \frac{\sum_{i=1}^\tau x_i^\gamma}{\norm{x}_0})^{\gamma^{-1}}$  for any $\gamma \in (0,1)$ by realizing that $
        \frac{\sum_{i=1}^\tau x_i^\gamma}{\tau} \geq \left(\prod_{i=1}^\tau x_i\right)^{\gamma \tau^{-1}}$.
\end{remark}

\subsection{$O(\epsilon)$-Mass Discount}\label{sec:discount}
\begin{theorem}
Given a state vector $x\in \Z^n$, without loss of generality, we assume $$|x_1|\geq |x_2|\geq\ldots \geq |x_\tau|>0=|x_{\tau+1}|=\ldots=|x_{n}|,$$ where $\tau = \norm{x}_0$. We define the $(\epsilon/2)$-mass discounted vector $x^{(\epsilon/2)}$ as follows.
\begin{align*}
    x^{(\epsilon/2)}_i &= 0, \quad  i\leq \epsilon \tau /2\\
    x^{(\epsilon/2)}_i &= x_i, \quad  i> \epsilon \tau /2.
\end{align*}

    Suppose $P$ is a random prime chosen by some scheme such that
    \begin{align*}
    \pr(\norm{x^{(\epsilon/2)}}_0 - \norm{x^{(\epsilon/2)}}_{0;\F_P}>\frac{\epsilon}{2}\norm{x^{(\epsilon/2)}}_0)&\leq \frac{1}{10},
    \intertext{then}
    \pr(\norm{x}_0 - \norm{x}_{0;\F_P}>\epsilon\norm{x}_0)&\leq \frac{1}{10}.
    \end{align*}    
\end{theorem}
\begin{proof}
By construction we have $\norm{x}_0 - \norm{x^{(\epsilon/2)}}_0=\frac{\epsilon}{2}\norm{x}_0$. Thus we have
    \begin{align*}
       & \pr(\norm{x}_0 - \norm{x}_{0;\F_P}>\epsilon\norm{x}_0)\\
        &= \pr(\norm{x^{(\epsilon/2)}}-\norm{x}_{0;\F_P}>\frac{\epsilon}{2}\norm{x}_0) 
        \intertext{where we always have $\norm{x}_{0;\F_P}\geq \norm{x^{(\epsilon/2)}}_{0;\F_P}$ and  $\norm{x}_{0}\geq \norm{x^{(\epsilon/2)}}_{0}$, so}    
        &\leq \pr(\norm{x^{(\epsilon/2)}}-\norm{x^{(\epsilon/2)}}_{0;\F_P}>\frac{\epsilon}{2}\norm{x^{(\epsilon/2)}}_0)\leq 1/10.
    \end{align*}
\end{proof}
This theorem formally proves that we may replace $x$ by $x^{(\epsilon/2)}$ in the upper bounds we derived earlier in this section.

\section*{Acknowledgement}
The author would like to thank Hung Ngo for raising the question of finding a counter part of \textsf{PCSA/HyperLogLog} that allows deletions and is simple to implement in practice. 
The author would like to thank Yaowei Long to be a supportive sanity-checker. 

\bibliographystyle{plain}
\bibliography{refs}

\appendix 

\section{Asymptotic Analysis}\label{app:math}
\begin{lemma}\label{lem:sub_exponential}
Let $X_r$ be a random variable with density $\nu(z,r)$ and $r\in[1/2,1]$. Then $X_r$ is sub-exponential. Specifically,
\begin{align*}
    \E 2^{|X_r|/6}  < 4.
\end{align*}
\end{lemma}
\begin{proof}
Note that
\begin{align*}
     \E 2^{|X_r|/6} = \int_{-\infty}^0 2^{-z/6}\nu(z,r)\,dz + \int_0^\infty 2^{z/6}\nu(z,r)\,dz.
\end{align*}
We estimate the two parts separately. Note that $\nu(z,r)\leq (1-e^{-2^{-z}})$ since $r\leq 1$.
We first estimate
\begin{align*}
    \int_0^\infty 2^{z/6}\nu(z,r)\,dz&\leq  \int_{0}^\infty 2^{z/6}(1-e^{-2^{-z}})\,dz < 2
\end{align*}
Note that since $r\in[1/2,1]$, $
\nu(z,r)\leq \prod_{j=1}^\infty (1-(1-e^{-2^{-z-j}})/2) $. Then we estimate
\begin{align*}
    \int_{-\infty}^0 2^{-z/6}\nu(z,r) \,dz \leq  \int_{-\infty}^0 2^{-z/6} \prod_{j=1}^\infty (1-(1-e^{-2^{-z-j}})/2) \,dz < 2
\end{align*}
We conclude that
\begin{align}
    \E 2^{|X_r|/6} < 2+2 =4,\label{eq:bound1}
\end{align}
from which we know $X_r$ is sub-exponential.
\end{proof}

\begin{lemma}\label{lem:taylor_expansion}
    Let $X_r$ be a random variable with density $\nu(z,r)$ and $r\in[1/2,1]$. Then
    \begin{align*}
         \log \E e^{tX_r}
    &= t\E X_r +\frac{t^2}{2}\var X_r + O(t^3),\quad t\to 0
    \end{align*}
\end{lemma}
\begin{proof}
Note that by the definition of the exponential function,
\begin{align*}
    \E e^{tX_r} = \sum_{k=0}^\infty \frac{t^k}{k!}\E X_r^k = 1 + t\E X_r +\frac{t^2}{2}\E X_r^2 + \sum_{k=3}^\infty \frac{t^k}{k!}\E X_r^k.
\end{align*}
We want to prove $\E e^{tX_r}=  1 + t\E X_r +\frac{t^2}{2}\E X_r^2 + O(t^3)$ as $t\to 0$. Then we need to bound the moments $\E |X_r|^k$.
Since $X_r$ is sub-exponential (Lemma \ref{lem:sub_exponential}), we can bound the moments in a standard way.
For $k\geq 1$,
    \begin{align}
        \E |X_r|^k&=\E (2^{ |X_r|/6}\cdot 2^{- |X_r|/6} |X_r|^k) = \int 2^{z/6} \cdot (2^{- z/6} z^k) \,d\mu_{|X_k|}\\
        &\leq \int 2^{z/6} \cdot \sup_{s\geq 0}(2^{- s/6}s^k) \,d\mu_{|X_k|}= \E 2^{ |X_r|/6} \sup_{s\geq 0}(2^{- s/6}s^k),\label{eq:moment}
    \end{align}
where $\sup_{s\geq 0}(2^{-s/6}s^k)=e^{-k}(6k/\log 2)^k$. Thus by Lemma \ref{lem:sub_exponential} and (\ref{eq:moment})
    \begin{align}
        \E |X_r|^k \leq 4 e^{-k}(6k/\log 2)^k < 4 \cdot3.2^k\cdot k^k,\label{eq:moment_bound}
    \end{align}
Thus we have
    \begin{align*}
        | \E e^{tX_r}-(1 + t\E X_r +\frac{t^2}{2}\E X_r^2)| &= \left|\sum_{k=3}^\infty \frac{t^k}{k!}\E X_r^k\right| \leq \sum_{k=3}^\infty \frac{t^k}{k!}\E |X_r|^k\\
        &\leq \sum_{k=3}^\infty \frac{t^k}{k!} 4 \cdot3.2^k\cdot k^k
    \end{align*}
    By Stirling's formula, we know for $k\in \N$,
    \begin{align*}
        k!> \sqrt{2\pi k}(k/e)^k>e^{-k} k^k.
    \end{align*}
    Thus
    \begin{align*}
        | \E e^{tX_r}-(1 + t\E X_r +\frac{t^2}{2}\E X_r^2)| &\leq 4\sum_{k=3}^\infty t^k (3.2e)^k\leq  8 t^3(3.2e)^3,
    \end{align*}
   as soon as $3.2et < 1/2$.
 Thus we have
\begin{align*}
   \E e^{tX_r}&= 1+t\E X_r +\frac{t^2}{2}\E X_r^2+O(t^3),\quad t\to 0.
    \intertext{Note that $\log(1+x) = x -\frac{1}{2}x^2 + O(x^3)$ as $x\to 0$. Since $X_r$ is sub-exponential (Lemma \ref{lem:sub_exponential}), both $|\E X_r|$ and $\E X^2_r$ are finite. Thus}
    \log  \E e^{tX_r}&= t\E X_r +\frac{t^2}{2}\E X_r^2-\frac{1}{2}t^2(\E X_r)^2 + O(t^3) \\
    &= t\E X_r +\frac{t^2}{2}\var X_r + O(t^3),\quad t\to 0,
\end{align*}
since $\var X_r = \E X_r^2 - (\E X_r)^2$.
\end{proof}

\begin{theorem}\label{thm:psi_asym}
Define $X_r$ to be a random variable with density $\nu(z,r)$ (see Definition \ref{def:nu}).  Then for $r\in[1/2,1]$,
\begin{align*}
    \phi(1/m,r)^{-m}&= 2^{-\E X_r}+O(m^{-1})\\
    \frac{\phi(2/m,r)^{m}}{\phi(1/m,r)^{2m}}  &= 1+(\log^2 2)m^{-1}\var X_r+O(m^{-2}),\quad m\to \infty.
\end{align*}
Furthermore, for any $r\in[1/2,1]$, $(\log^2 2)\var X_r<79$.
We define $\psi_\E(|F|) = \E X_{1-|F|^{-1}}$ and $\psi_\var(|F|) = \var X_{1-|F|^{-1}}$.
\end{theorem}
\begin{proof}
By construction, we have $\phi(t,r)=\int_{-\infty}^\infty 2^{t z}\nu(z,r)\,dz = \E 2^{tX_r}$. Then we have
\begin{align*}
    \log (\phi(1/m,r)) &= \log \E 2^{X_r/m} = \log \E e^{(\log 2)X_r/m} 
    \intertext{by Lemma \ref{lem:taylor_expansion}, as $m\to \infty$,}
    &=((\log 2)/m)\E X_r +\frac{((\log 2)/m)^2}{2}\var X_r + O(m^{-3}).
\end{align*}
Thus we have
    \begin{align*}    
    \log (\phi(1/m,r)^{-m})&= (-m)  \log (\phi(1/m,r))\\
    &=(-m)\left(( (\log 2)/m)\E X_r +\frac{(\log (2)/m)^2}{2}\var X_r + O(m^{-3})\right)\\
    &=     - (\log 2)\E X_r + O(m^{-1}).\\
        \log\left(\frac{\phi(2/m,r)^{m}}{\phi(1/m,r)^{2m}} \right)
    &=m \log (\phi(2/m,r)) - 2m  \log (\phi(1/m,r)) \\ 
    &=m\left(( 2(\log 2)/m)\E X_r +\frac{(2\log (2)/m)^2}{2}\var X_r + O(m^{-3})\right)\\
        &-2m\left(( (\log 2)/m)\E X_r +\frac{(\log (2)/m)^2}{2}\var X_r + O(m^{-3})\right)\\
    &= (\log^2 2)m^{-1}\var X_r +O(m^{-2}).
    \end{align*}
Note that $e^x = 1+x + O(x^2)$ as $x\to 0$. Thus
\begin{align*}
    \phi(1/m,r)^{-m} &= e^{  - (\log 2)\E X_r + O(m^{-1})} =   e^{- (\log 2)\E X_r} (1+O(m^{-1}))=2^{-\E X_r}+O(m^{-1}).\\
    \frac{\phi(2/m,r)^{m}}{\phi(1/m,r)^{2m}}& = e^{(\log^2 2)m^{-1}\var X_r +O(m^{-2})} = 1 + (\log^2 2)m^{-1}\var X_r + O(m^{-2}).
\end{align*}
Finally, to bound $\var X_r$, note that
\begin{align*}
    \var X_r &= \E X_r^2 - (\E X_r)^2 \leq \E X_r^2 
    \intertext{by the moment bound (\ref{eq:moment_bound})}
    \E X_r^2& \leq 4 \cdot 3.2^2 \cdot 2^2 < 164.
\end{align*}
We conclude that $(\log^2 2)\var X_r < 79$.
\end{proof}

\end{document}